\documentclass[runningheads]{llncs}
\usepackage{amssymb}
\usepackage{amsmath}
\usepackage{bussproofs}
\usepackage[noend]{algpseudocode}
\usepackage{algorithm}
\usepackage{todonotes}
\usepackage{booktabs}
\usepackage{graphicx}
\usepackage{thmtools}
\usepackage{thm-restate}

\usepackage{caption}
\usepackage[inline]{enumitem}

\renewenvironment{proof}{\begin{pf}}{\qed\end{pf}}

\usepackage{url}

\usepackage{todonotes}

\newcommand\restr[2]{{
  \left.\kern-\nulldelimiterspace 
  #1 
  \right|_{#2} 
  }}

\newcommand*{\defeq}{\mathrel{\vcenter{\baselineskip0.5ex \lineskiplimit0pt
                     \hbox{\scriptsize.}\hbox{\scriptsize.}}}%
                     =}     
   
\newcommand{\funDom}{\ensuremath{\mathbf{dom}}}

\newcommand{\FFF}{\ensuremath{\Phi}}

\newcommand{\qprefix}{\mathcal{Q}}
\newcommand{\matrixf}{\varphi}
\newcommand{\var}{\mathit{var}}

\newcommand{\tool}{{\sc Pedant}}%

\algnewcommand{\LineComment}[1]{\State \(\triangleright\) \parbox[t]{0.8\linewidth}{#1}}

\algdef{SN}{Loop}{EndLoop}[0]{\textbf{loop}}

\algdef{SN}{Vars}{EndVars}[0]{\textbf{Global}}
\MakeRobust{\Call}

\begin{document}
\title{Certified DQBF Solving by Definition Extraction}
\author{Franz-Xaver Reichl \and Friedrich Slivovsky \and Stefan Szeider}
\authorrunning{F. Reichl et al.}
\institute{
TU Wien, Vienna, Austria\\
\email{\{freichl,fs,sz\}@ac.tuwien.ac.at}
}
\maketitle              
\begin{abstract}
  We propose a new decision procedure for dependency quantified Boolean formulas (DQBFs) that uses interpolation-based definition extraction to compute Skolem functions in a counter-example guided inductive synthesis (CEGIS) loop.
In each iteration, a family of candidate Skolem functions is tested for correctness using a SAT solver, which either determines that a model has been found, or returns an assignment of the universal variables as a counterexample.
Fixing a counterexample generally involves changing candidates of multiple existential variables with incomparable dependency sets.
 Our procedure introduces auxiliary variables---which we call
 \emph{arbiter variables}---that each represent
 the value of an
 existential variable for a particular assignment of its dependency
 set.
  Possible repairs are expressed as clauses on these variables, and a SAT solver is invoked to find an assignment that deals with all previously seen counterexamples.
Arbiter variables define the values of Skolem functions for assignments where they were previously undefined, and may lead to the detection of further Skolem functions by definition extraction.

A key feature of the proposed procedure is that it is certifying by design:
for true DQBF, models can be returned at minimal overhead.
Towards certification of false formulas, we prove that clauses can be derived in an expansion-based proof system for DQBF.
 
In an experimental evaluation on standard benchmark sets, a prototype implementation was able to match (and in some cases, surpass) the performance of state-of-the-art-solvers.
  Moreover, models could be extracted and validated for all true instances that were solved.%
\end{abstract}
%
\section{Introduction}
Sustained progress in propositional satisfiability (SAT) solving~\cite{HeuleJS19} has resulted in a growing number of applications in the area of electronic design automation~\cite{VizelWM15}, such as model checking~\cite{BiereCCZ99}, synthesis~\cite{Solar-LezamaTBSS06}, and symbolic execution~\cite{BaldoniCDDF18}.
Efficient SAT solvers were essential for recent progress in constrained sampling and counting~\cite{MeelVCFSFIM16}, two problems with many applications in artificial intelligence.
In these cases, SAT solvers are used to deal with problems from complexity classes beyond NP and propositional encodings that grow super-polynomially in the size of the original instances.
As a consequence, these problems are not directly encoded in propositional logic but have to be reduced to a sequence of SAT instances.

The success of SAT solving on the one hand, and the inability of propositional logic to succinctly encode problems of interest on the other hand, have prompted the development of decision procedures for more succinct generalizations of propositional logic such as Quantified Boolean Formulas (QBFs).
Evaluating QBFs is PSPACE-complete~\cite{StockmeyerM73} and thus believed to be much harder than SAT, but in practice the benefits of a smaller encoding may outweigh the disadvantage of slower decision procedures~\cite{FaymonvilleFRT17}. A QBF is true if it has a \emph{model}, which is a family of Boolean functions (often called \emph{Skolem functions}) that satisfy the matrix of the input formula for each assignment of universal variables.
The arguments of each Skolem function are implicitly determined by the nesting of existential and universal quantifiers.
Dependency QBF (DQBF) explicitly state a \emph{dependency set} for each existential variable, which is a subset of universal variables allowed as arguments of the corresponding Skolem function~\cite{AzharPR01,BalabanovCJ14}.
As such, they can succinctly encode the existence of Boolean functions subject to a set of constraints~\cite{Rabe17}, and problems like equivalence checking of partial circuit designs~\cite{GitinaRSWSB13} and bounded synthesis~\cite{FaymonvilleFRT17} can be naturally expressed in this way.

Several decision procedures for DQBF have been developed in recent years (see Section~\ref{sec:related}).
Conceptually, these solvers either reduce to SAT or QBF by instantiating~\cite{FrohlichKBV14} or eliminating universal variables~\cite{GitinaWRSSB15,WimmerKBS17,Ge-ErnstSW19,Sic20}, or lift Conflict-Driven Clause Learning (CDCL) to non-linear quantifier prefixes by imposing additional constraints~\cite{FrohlichKB12,TentrupR19}.\footnote{An approach that does not fit this simplified classification is the First-Order solver {\sc iProver}~\cite{Korovin08}.}
We believe these methods should be complemented with algorithms that directly reason at the level of Skolem functions~\cite{RabeS16}.
A strong argument in favor of such an approach is the fact that DQBF instances often have a large fraction of unique Skolem functions that can be obtained by definition extraction, but the current solving paradigms have no direct way of exploiting this~\cite{Slivovsky20}.

In this paper, we develop new decision procedures for DQBF designed around computing Skolem functions by definition extraction.
We first describe a simple algorithm that proceeds in two phases. In the first phase, it introduces clauses to make sure each existential variable is defined in terms of its dependency set and auxiliary \emph{arbiter variables}. In the second phase, it searches for an assignment of the arbiter variables under which the definitions are a model.
Runs of this algorithm can degenerate into an exhaustive instantiation of dependency sets for easy cases, so we propose an improved version in the Counter-Example Guided Inductive Synthesis (CEGIS) paradigm~\cite{Solar-LezamaTBSS06,Solar-LezamaJB08,JhaS17}.

We implemented the CEGIS algorithm in a system named \tool.
In an experimental evaluation, \tool\ performs very well compared to a selection of state-of-the-art solvers---notably, it achieves good performance without the aid of the powerful preprocessor {\sc HQSPre}~\cite{WimmerSB19}.
One of the benefits of its function-centric design is that \tool\ internally computes a family of Skolem functions, and can output models of true instances at a negligible overhead.
Using a simple workflow, we are able to validate models for all true instances solved by \tool.
Towards validation of false instances, we prove that clauses introduced by \tool\ can be derived in the $\forall$Exp+Res proof system~\cite{JanotaM13,BeyersdorffBCSS19}.

The remainder of the paper is structured as follows.
After covering basic concepts in Section~\ref{sec:preliminaries}, we present the new decision algorithms for DQBF and prove their correctness in Section~\ref{sec:algorithms}.
We describe the implementation and experimental results in Section~\ref{sec:experiments}.
We discuss related work in Section~\ref{sec:related}, before concluding with an outlook on future work in Section~\ref{sec:conclusion}.
\section{Preliminaries}\label{sec:preliminaries}
\paragraph{Propositional Logic} 
A \emph{literal} is either a variable or the negation of a variable. 
A \emph{clause} is the disjunction of literals. 
A \emph{term} is a conjunction of literals.
A formula is in \emph{Conjunctive Normal Form} (CNF) if it is a conjunction of clauses.
Whenever convenient, we identify a CNF with a set of clauses and clauses, respectively terms, with sets of literals.
We denote the set of variables occurring in a formula $\varphi$ by $\mathit{var}(\varphi)$.
We denote the truth value \emph{true} by \textsc{true} and \emph{false} by \textsc{false}.
An \emph{assignment} of a set $V$ of variables is a function mapping $V$ to $\{\textsc{true},\textsc{false}\}$. 
We denote the set of all assignments for $V$ by $[V]$.
Moreover, we associate an assignment $\sigma$ with the term 
$\{x\mid x\in\funDom(\sigma),\sigma(x)=\textsc{true}\}\cup\{\neg x\mid x\in\funDom(\sigma),\sigma(x)=\textsc{false}\}$. 
Whenever convenient, we treat assignments as terms. 
Let $\sigma\in[V]$ and let $W\subseteq V$, then we denote the \emph{restriction of $\sigma$ to $W$} by $\restr{\sigma}{W}$.
For a formula $\varphi$ and an assignment $\sigma$ we denote the
\emph{evaluation of $\varphi$ by $\sigma$} with $\varphi[\sigma]$. 
A formula $\varphi$ is \emph{satisfied} by an assignment $\sigma$ if $\varphi[\sigma]=\textsc{true}$ and it is \emph{falsified} by $\sigma$ otherwise.
A formula $\varphi$ is \emph{satisfiable} if there is an assignment $\sigma$ that satisfies $\varphi$ and it
is \emph{unsatisfiable} otherwise.
Let $\varphi$ and $\psi$ be two formulae, $\varphi$ \emph{entails} $\psi$, denoted by $\varphi\vDash\psi$, if every
assignment satisfying $\varphi$ also satisfies $\psi$.
A \emph{definition} for a variable $x$ by a set of variables $X$ in a formula $\varphi$ is a formula $\psi$ with $\mathit{var}(\psi)\subseteq X$ such that for every satisfying assignment $\sigma$ of $\varphi$ the equality $\sigma(x)=\psi[\sigma]$ holds~\cite{Slivovsky20}.

\paragraph{Dependency Quantified Boolean formulas}
We only consider \emph{Dependency Quantified Boolean formulas} (DQBF) in \emph{Prenex Conjunctive Normal Form}
(PCNF). A DQBF in PCNF is denoted by $\Phi=\qprefix.\varphi$, where
\begin {enumerate*} [label=(\itshape\alph*\upshape)]
	\item the \emph{quantifier prefix} $\qprefix$ is given by
	$\qprefix=\forall u_1\ldots\forall u_n\exists e_1(D_1)\ldots\exists e_m(D_m)$. 
	Here $u_1,\ldots,u_n$ and $e_1,\ldots,e_m$ shall be pairwise different variables.
	We denote the set $\{u_1,\ldots,u_n\}$ by $U_{\Phi}$ and the set $\{e_1,\ldots,e_m\}$ by $E_{\Phi}$.
	Additionally, $D_1,\ldots,D_m$ shall be subsets of $U_{\Phi}$.
	\item the \emph{matrix} $\matrixf$ shall be a CNF with 
	$\mathit{var}(\matrixf)\subseteq U_\Phi\cup E_\Phi$.
\end {enumerate*} 
For $1\leq i\leq m$ we call the set $D_i$ the \emph{dependencies} of $e_i$.
We refer to the variables in $U_\Phi$ as \emph{universal variables} 
and to the variables in $E_\Phi$ as \emph{existential variables}.
For an existential variable $e$ we denote its dependencies by $D_\Phi(e)$.
If the underlying DQBF is clear from the context we omit the subscript.

Let $\Phi$ be a DQBF and $F$ be a set of functions $\{f_{e_1},\ldots,f_{e_m}\}$ such that for $1\leq i\leq m$,
$f_{e_i}: [D_i]\rightarrow\{\textsc{true},\textsc{false}\}$. 
For an assignment $\sigma$ to the universal variables we denote the existential assignment 
$\{f_{e_1}(\restr{\sigma}{D_1}),\ldots,f_{e_m}(\restr{\sigma}{D_m})\}$ by $F(\sigma)$.
$F$ is a \emph{model} (or a \emph{winning $\exists$-strategy}) for $\Phi$ if for each
assignment $\sigma$ to the universal variables, the assignment $\sigma\cup F(\sigma)$ satisfies the matrix $\matrixf$. A DQBF is true if it has a model and false otherwise.

\paragraph{$\forall$Exp+Res}
The \emph{DQBF-$\forall$Exp+Res}~\cite{BeyersdorffBCSS19} calculus is a proof system for DQBF, which is based on the \emph{$\forall$Exp+Res} calculus for QBF. 
It instantiates the matrix of a DQBF with a universal assignment and uses propositional resolution on the instantiated clauses. 
This proof system is sound and refutationally complete~\cite{BeyersdorffBCSS19}.
Since we are interested in DQBF, we refer to DQBF-$\forall$Exp+Res simply as $\forall$Exp+Res.
The rules for the system are given in Fig.~\ref{ForallExp}.

\noindent\fbox{\begin{minipage}{\linewidth}
\begin{prooftree}
\RightLabel{(axiom)}
\AxiomC{}
\UnaryInfC{$\{\ell^{\restr{\sigma}{D(\var(\ell))}}\mid \ell\in C,\mathit{var}(\ell)\in E\}$}
\end{prooftree}
Where $C$ is a clause in the matrix of the DQBF and $\sigma$ is an assignment to the universal variables
that falsifies each universal literal in $C$.
Note that variables with a different annotation denote different variables.

The second rule is the propositional resolution rule.

\begin{prooftree}
\RightLabel{(resolution)}
\AxiomC{$C_1\cup\{x^\tau\}$}
\AxiomC{$C_2\cup\{\neg x^\tau\}$}
\BinaryInfC{$C_1\cup C_2$}
\end{prooftree}
Where $x$ is an existential variable, $\tau$ an assignment for the universal variables in $D(e)$ and where $C_1$ and $C_2$ are clauses.
\captionof{figure}{The rules of DQBF-$\forall$Exp+Res}
\label{ForallExp}
\end{minipage}}

\section{Solving DQBF by Definition Extraction}\label{sec:algorithms}
In this section, we describe two decision procedures for DQBF that leverage definition extraction.
We start with an algorithm (Algorithm~\ref{alg:BasicAlgorithm}) that is fairly simple but introduces some important concepts.
Because this algorithm leads to the equivalent of exhaustive expansion of universal variables on trivial examples, we then introduce a more sophisticated algorithm based on CEGIS (Algorithm~\ref{alg:AdvancedAlgorithm}).
We also sketch correctness proofs for both algorithms.

Throughout this section, we consider a fixed DQBF $\FFF \defeq \qprefix.\;\matrixf$ with quantifier prefix $\qprefix \defeq \forall u_1\ldots\forall u_n\allowbreak \exists e_1(D_1)\ldots\exists e_m(D_m)$.
\subsection{A Two-Phase Algorithm}\label{Section:BasicAlgorithm}
The algorithm proceeds in two phases. In the first phase ({\sc GenerateDefinitions}), it finds definitions $\psi_\mathit{Def}$ for all existential variables.
It maintains a set $A$ of auxiliary \emph{arbiter variables} whose semantics are encoded in a set $\varphi_A$ of \emph{arbiter clauses}, both of which are empty initially.
If a variable $e_i$ is defined in terms of its dependency set, the definition is computed using a SAT solver (line~\ref{alg1:getdefinition}) capable of generating interpolants~\cite{Slivovsky20}.
Otherwise, the SAT solver returns an assignment~$\xi$ of the dependency set $D_i$ and the arbiter variables $A$ for which the variable is not defined.
In particular, $e_i$ is not defined under the restriction $\sigma = \xi|_{D_i}$ to its dependency set.
The algorithm then introduces an arbiter variable $e_i^{\sigma}$ that determines the value of the Skolem function for $e_i$ under $\sigma$.
In subsequent iterations, we include these arbiter variables in the set of variables that can be used in a definition of $e_i$.
The newly introduced clauses ensure that $e_i$ and $e_i^{\sigma}$ take the same value under the assignment $\sigma$ (line~\ref{alg1:newarbiterclauses}), so that $e_i$ is defined by $e_i^{\sigma}$ and~$D_i$.
Since the number of assignments $\sigma$ of the dependency set is bounded, we will eventually find a definition of $e_i$ in terms of its dependency set $D_i$ and the arbiter variables~$A$.

In the second phase ({\sc FindArbiterAssignment}), a SAT solver (line~\ref{alg1:arbitersolver}) is used to find an assignment of the arbiter variables under which the definitions obtained in the first phase are a model.
Starting with an initial assignment~$\tau$, we use a SAT solver
to check whether the formula $\psi_\mathit{Def} \land \neg \varphi$ consisting of the definitions from the first phase and the negated matrix of the input DQBF is unsatisfiable under~$\tau$ (line~\ref{alg1:validitycheck}).
If that is the case, Algorithm~\ref{alg:BasicAlgorithm} returns {\sc true}.
Otherwise, the SAT solver returns an assignment $\sigma$ as a counterexample.
Since the existential variables are defined in $\varphi \land \varphi_A$ by the universal and arbiter variables, the formula $\varphi \land \varphi_A$ must be unsatisfiable under the assignment $\tau \land \sigma|_U$ consisting of the arbiter assignment and counterexample restricted to universal variables.
A \emph{core} $\rho$ of failed assumptions $\tau \land \sigma_U$ such that $\rho \models \neg (\varphi \land \varphi_A)$ is extracted using another SAT call.
The assignment $\rho|_A$ represents a concise reason for the failure of the arbiter assignment $\tau$, and its negation $\neg \rho|_A$ is added as a new clause to the SAT solver used to generate arbiter assignments, which is subsequently invoked to find a new arbiter assignment.

This process continues until a model is found or the SAT solver cannot find a new arbiter assignment, in which case the algorithm returns {\sc false}.

\begin{algorithm}[ht!]
\caption{Solving DQBF by  Definition Extraction}\label{alg:BasicAlgorithm}
\begin{algorithmic}[1]
	\vspace{0.5em}
	\Procedure{SolveByDefinitionExtraction}{$\FFF$}
	\State$(\varphi_A,A,\psi_\mathit{Def})\gets\Call{GenerateDefinitions}{\FFF}$
	\State\Return$\Call{FindArbiterAssignment}{\FFF,\varphi_A,A,\psi_\mathit{Def}}$
    \EndProcedure
%
%
\Statex
\Procedure{GenerateDefinitions}{$\FFF$}
	\LineComment{$\FFF=\forall u_1\ldots\forall u_n\exists e_1(D_1)\ldots\exists e_m(D_m).\matrixf$}
	
	\State $A\gets\emptyset,\;\psi_\mathit{Def}\gets\emptyset,\;\varphi_A\gets\emptyset$
        \LineComment{$A$: arbiter variables, $\psi_\mathit{Def}$: definitions, $\varphi_A$: arbiter clauses}
	\For{$i = 1,\dots,m$}
		\State $\mathit{isDefined},\xi \gets\textsc{isDefined}
			(e_i,A\cup D_i,\varphi\wedge\varphi_A)$
		\While{not $\mathit{isDefined}$}
			\State $\sigma\gets\restr{\xi}{D_i}$ \Comment{$e_i$ is not defined under $\xi \in [D_i \cup A]$}
			\State $A\gets A\cup\{e^{\sigma}_i\}$
			\State $\varphi_A\gets\varphi_A\wedge(e^{\sigma}_i\vee\neg\sigma\vee\neg e_i)\wedge	(\neg e^{\sigma}_i\vee\neg\sigma\vee e_i)$ \label{alg1:newarbiterclauses}
			\State $\mathit{isDefined},\xi\gets\textsc{isDefined}
			(e_i,A\cup D_i,\varphi\wedge\varphi_A)$
		\EndWhile
		\State $\psi_\mathit{Def}^i\gets\Call{getDefinition}{e_i,A\cup D_i,\varphi\wedge\varphi_A}$ \label{alg1:getdefinition}
		\State $\psi_\mathit{Def}\gets\psi_\mathit{Def}\wedge(e_i\leftrightarrow\psi_\mathit{Def}^i)$
	\EndFor
	\State\Return $(\varphi_A,A,\psi_\mathit{Def})$
\EndProcedure
%
\Statex
\Procedure{FindArbiterAssignment}{$\FFF,\varphi_A,A,\psi_\mathit{Def}$}
	\State $\tau\gets\bigwedge_{a\in A} a$
	\Comment initial assignment to the arbiter variables
	\State $\mathit{validitySolver}\gets\textsc{SatSolver}(\mathit{\psi_\mathit{Def}}\wedge\neg\matrixf)$
	\State $\mathit{arbiterSolver}\gets\textsc{SatSolver}(\emptyset)$ \label{alg1:arbitersolver}
	\Loop
		\If{$\mathit{validitySolver}.\Call{solve}{\tau}$} \label{alg1:validitycheck}
			\State $\sigma\gets\mathit{validitySolver}.\Call{getModel}$
				\State $\rho\gets\Call{getCore}{\varphi\wedge\varphi_A,\tau\wedge\restr{\sigma}{U}}$ \label{alg1:getCore}	
				\State $\mathit{arbiterSolver}.\Call{addClause}{\neg\restr{\rho}{A}}$ \label{alg1:addclause}	
				\If{$\mathit{arbiterSolver}.\textsc{solve}()$}
					\State$\tau\gets\mathit{arbiterSolver}.\Call{getModel}$
				\Else
					\State\Return {\sc false}
				\EndIf	
		\Else
			\State \Return {\sc true}
		\EndIf		
	\EndLoop
\EndProcedure
\end{algorithmic}
\end{algorithm}
We now argue that Algorithm~\ref{alg:BasicAlgorithm} is a decision procedure for DQBF.
In the following, $A$ shall denote a set of arbiter variables and $\varphi_A$ shall denote the associated set of arbiter clauses.
We will first give two auxiliary properties of definitions, which we will use in subsequent proofs.
\begin{lemma}\label{DefinitionsArePreserved}
Let $\varphi$ and $\xi$ be two propositional formulas, let $e\in\var(\varphi)$ and $S\subseteq\var(\varphi)$.
Moreover, let $\psi_\mathit{Def}$ be a definition for $e$ by $S$ in $\varphi$. Then $\psi_\mathit{Def}$ is also a definition for $e$ by $S$ in $\varphi\wedge\xi$.
\end{lemma}
\begin{proof}
Obviously any satisfying assignment $\sigma$ of $\varphi\wedge\xi$ also satisfies $\varphi$. 
By the premise of the lemma we know that for every satisfying assignment $\sigma$ of $\varphi$ we have
$\sigma(e)=\psi_\mathit{Def}[\sigma]$. By combining the above two properties we get that $\sigma(e)=\psi_\mathit{Def}[\sigma]$ holds for
every satisfying assignment for $\varphi\wedge\xi$. This proves the lemma.
\end{proof}

\begin{remark}
The above Lemma implies that if we once find a definition in Algorithm~\ref{alg:BasicAlgorithm} then we do not destroy it by adding additional clauses.
\end{remark}

\begin{lemma}\label{DefinitionRemoval}
Let $\psi$ be a formula and let $S\subseteq\var(\psi)$ such that each $s\in S$ has a definition $\psi_s$ by a set 
$D_s\subseteq\var(\psi)$ in $\psi$. Then $\psi$ and $\psi\wedge\bigwedge_{v\in S}(v\leftrightarrow\psi_s)$ are equivalent.
\end{lemma}
\begin{proof}
First assume that $\psi$ is satisfied by an assignment $\sigma$. By the nature of definition we have $\sigma(v)=\psi_v[\sigma]$ for each $v\in S$. Thus, $\sigma$ satisfies $\psi\wedge\bigwedge_{v\in S}(v\leftrightarrow\psi_s)$. The other direction of the equivalence is obvious.
\end{proof}

A DQBF has a model if, and only if, there is a propositional formula for each existential variable that defines its Skolem functions using only variables from the dependency set.
This can be slightly generalized by allowing the definition to contain existential variables whose dependency sets are a subset.
\begin{lemma}
\label{TrueDQBFAux}
Let $\FFF$ be a DQBF and $<_E$ a linear ordering of its existential variables.
Then $\FFF$ is true if, and only if, for each $e\in E$ there is a formula $\psi_e$ with $\var(\psi_e)\subseteq D(e)\cup\{x\in E\mid D(x)\subseteq D(e),x<_E e\}$ such that 
$\neg\varphi\wedge\bigwedge_{e\in E}(e\leftrightarrow\psi_e)$ is unsatisfiable. 
\end{lemma}

\begin{proof}
First we assume that $\Phi$ is true. This means that there is a model $F=\{f_1,\ldots f_m\}$. 
We can encode these functions by formulas $\psi_e$. 
Obviously we have for each $e\in E$ that $\var(\psi_e)\subseteq D(e)$---as the model function
for $e$ may only depend on $D(e)$.
Now assume that $\neg\varphi\wedge\bigwedge_{e\in E}(e\leftrightarrow\psi_e)$ is satisfiable, i.e.,\ there is an assignment $\sigma$ that satisfies the formula. But this implies that there is an assignment to the universal variables such that the matrix is falsified under the model. 
As this contradicts our initial assumption, we know that the formula is unsatisfiable.

Next we assume the other side of the equivalence and  show
how we can construct a model for $\Phi$.
For this purpose we distinguish between two cases. 
First let $e$ be an existential variable such that $\var(\psi_e)\subseteq D(e)$ and
let $\sigma$ be an assignment to $D(e)$. Then we define $f_e(\sigma)\defeq\psi_e[\sigma]$.
Now let $e$ be the smallest variable (with respect to $<_E$) in $E$ such that $\var(\psi_e)\nsubseteq D(e)$. 
By using the assumption it follows that for each variable
$x\in\var(\psi_e)\setminus D(e)$ we have $x<_E e$.
Because of the minimality of $e$, this means that for each $x$ in $\var(\psi_e)\setminus D(e)$
we have already defined a function $f_x$.
Now we define for $\sigma\in[D(e)]$ the assignment
$\rho(\sigma)\defeq\{(x,f_x(\restr{\sigma}{D(x)}))\mid x\in\var(\psi_e)\setminus D(e)\}.$
Finally, we define $f_e(\sigma)\defeq\psi_e[\rho(\sigma)]$---the construction guarantees that $\var(\psi_e[\rho(\sigma)])\subseteq D(e)$.

For the remaining variables we proceed inductively. We denote the resulting set of functions by $F$.
It remains to show that $F$ is actually a model. We assume the opposite, i.e., there is a universal assignment $\sigma$ such that $\matrixf$ is falsified by $\sigma\cup F(\sigma)$.
It can be proven that $\sigma\cup F(\sigma)$ satisfies $\bigwedge_{e\in E}(e\leftrightarrow\psi_e)$.
As $\sigma\cup F(\sigma)$ falsifies the matrix $\matrixf$ we get a contradiction to our initial assumption.

\end{proof}

\begin{theorem}\label{basicalg:cor}
If Algorithm~\ref{alg:BasicAlgorithm} returns {\sc true} for the DQBF $\Phi$ then $\Phi$ is true.
\end{theorem}
\begin{proof}
Let $\FFF'\defeq\qprefix\exists A(\emptyset).\matrixf$, and let $<_E$ be any ordering of existential variables in~$\FFF'$ in which the variables in $A$ come before the remaining variables.
If Algorithm~\ref{alg:BasicAlgorithm} returns {\sc true},  we know that there is an arbiter assignment $\tau$ such that $\neg\matrixf\wedge\psi_\mathit{Def}\wedge\tau$ is unsatisfiable.
For each arbiter variable $e^\sigma$, we obtain a definition $\psi_e^\sigma$ as $\psi_e^\sigma \defeq \tau(e^\sigma)$. 
We can now replace the arbiter assignment $\tau$ with these definitions and apply Lemma~\ref{TrueDQBFAux} to conclude that $\FFF'$ is true.
But if $\FFF'$ is true, then necessarily also $\FFF$ is true.
\end{proof}
To show that the algorithm returns {\sc false} only if the input DQBF is false, one can prove that clauses on arbiter variables introduced by \Call{FindArbiterAssignment}{} can be derived (as clauses on annotated literals) in $\forall$Exp+Res.

\begin{definition}
Let $\tau$ be a (partial) assignment to the arbiter variables and let $\rho$ be an assignment to the universal variables. We define the term 
$\tau^\rho\defeq\{\ell^\sigma\in\tau\mid\rho\vDash\sigma\}$
\end{definition}
\begin{definition}
Let $\sigma$ be an assignment to the universal variables.
We denote the formula that is the result of instantiating $\matrixf$ by $\sigma$ in the sense of the $\forall$Exp+Res calculus by $\matrixf^\rho$. This means:
\[\matrixf^\sigma\defeq\{\{x^{\restr{\sigma}{D(\var(x))}}\mid x\in C,\var(x)\in E\}\mid C\in\matrixf,\forall \ell\in C:\var(\ell)\in U\Rightarrow\sigma\nvDash\ell\}\]
\end{definition}
\begin{remark}
Subsequently, we say that an arbiter clause $C=\ell^\sigma\vee\neg\sigma\vee\neg \ell$ 
is active with respect to an assignment $\rho$ if
$\rho$ falsifies $\ell^\sigma$ and $\rho$ satisfies $\sigma$.
\end{remark}%

\begin{lemma}\label{CorrespondanceArbiterAssignmentAndAnnotatedLiterals}
Let $\tau$ be a (partial) arbiter assignment and $\rho$ a universal assignment. Then $\matrixf\wedge\varphi_A\wedge\tau\wedge\rho$ and $\matrixf^\rho\wedge\tau^\rho$ are equisatisfiable.
\end{lemma}
\begin{proof}
First we assume that $\varphi\wedge\varphi_A\wedge\tau\wedge\rho$ is satisfied by an assignment $\lambda$.
We now construct a satisfying assignment for $\matrixf^\rho\wedge\tau^\rho$.
We define an assignment $\mu$ as $\mu(e^{\restr{\rho}{D(e)}})\defeq\lambda(e)$ for $e\in E$.
It follows by the definition of $\matrixf^\rho$, respectively of $\mu$ that $\mu$ satisfies $\matrixf^\rho$.
Now let $\ell^{\sigma}\in\tau^\rho$ be arbitrary but fixed. 
We can conclude that there is an arbiter clause 
$\neg \ell^{\sigma}\vee\neg\sigma\vee \ell$ that is active with respect to $\tau\cup\rho$.
As the clause is active, we know that $\lambda\vDash\ell$. 
This means that $\mu\vDash\ell^\sigma$. As the literal was arbitrary, $\mu$ satisfies $\tau^\rho$

Now assume that $\varphi\wedge\varphi_A\wedge\tau\wedge\rho$ is unsatisfiable and show that
$\matrixf^\rho\wedge\tau^\rho$ is unsatisfiable. 
For this purpose we assume that $\matrixf^\rho\wedge\tau^\rho$  is satisfied by an assignment $\lambda$ and show that this implies that we can construct a satisfying assignment $\mu$ for $\varphi\wedge\varphi_A\wedge\tau\wedge\rho$.
We define $\mu$ as:
\[
	\mu(x)=
	\begin{cases}
		\rho(x)			& x\in U\\
		\lambda(x^{\sigma})	& x\in E \wedge x^{\sigma}\in\var(\varphi^\rho\wedge\tau^\rho)\\
		\tau(x)			& x\in\funDom(\tau)\\
		\textsc{true}			& \text{otherwise}
	\end{cases}
\]
This assignment is well-defined as for each $e\in E$ there is maximally one annotation.
Obviously $\mu$ satisfies $\varphi$ - any clause containing a universal literal in the same polarity
as in $\rho$ is trivially satisfied, the remaining clauses are satisfied as $\lambda$ satisfies
$\varphi^\rho$. We now show that $\mu$ also satisfies ${\varphi}_A$. 
For this purpose let $\ell^{\sigma}\vee\neg\sigma\vee\neg \ell$ be a clause in ${\varphi}_A$.
If $\mu$ satisfies $\ell^{\sigma}$ or $\neg\sigma$ the clause is satisfied, 
so assume the opposite---i.e.\ $\mu\vDash\neg \ell^{\sigma}$ and $\mu\vDash\sigma$.
This means that we have $\rho\vDash\sigma$. Now we have to differentiate between three cases:
If $\var(\ell^\sigma)\in\funDom(\tau)$ then we have $\tau\vDash\ell^\sigma$, which implies $\tau^\sigma\vDash\ell^\sigma$.
Thus, as $\lambda$ satisfies $\tau^\sigma$ we know that $\mu\vDash\neg \ell$.
If $\var(\ell^\sigma)\notin\funDom(\tau)$ and $\var(\ell^\sigma)\in\var(\varphi^\rho\wedge\tau^\rho)$ then $\mu(\ell)=\mu(\ell^\sigma)$.
If $\var(\ell^\sigma)\notin\funDom(\tau)$ and $\var(\ell^\sigma)\notin\var(\varphi^\rho\wedge\tau^\rho)$ we have $\mu(\var(\ell))=\mu(var(\ell^\sigma))=\textsc{true}$.
To sum up this means that $\mu$ satisfies the arbiter clause.
As the arbiter clause was arbitrary we thus know that $\mu$ satisfies $\varphi\wedge{\varphi}_A\wedge\tau\wedge\rho$. 
But this is a contradiction. Thus, we have shown that $\varphi^\rho\wedge\tau^\rho$ is unsatisfiable.
\end{proof}

\begin{remark}
Let $\pi$ be a refutation of a formula $\varphi$. Additionally, let $C_1,C_2$ and $C$ be clauses in $\pi$ such
that $C$ is the result of resolving $C_1$ and $C_2$ with respect to a pivot $x$. We denote the clause $C$ by
$\mathcal{R}_x(C_1,C_2)$.
\end{remark}

\begin{lemma}\label{Aux:PropositionalResult}
Let $\varphi$ be a formula and $L=\{\ell_1,\cdots,\ell_n\}$ be a set of literals with pairwise different variables.
If $\varphi\wedge L$ is unsatisfiable and for each $\hat{L}\subsetneq L$ the formula $\varphi\wedge\hat{L}$ is satisfiable then the clause $\neg \ell_1\vee\ldots\vee\neg \ell_n$ can be derived from $\varphi$ by resolution.
\end{lemma}
\begin{proof}
Let $\hat{\varphi}$ be the clause which is the result of removing all clauses from $\varphi$ which are subsumed by a literal in $L$. 
Now we can conclude the following properties:
\begin{itemize}
	\item $\hat{\varphi}\wedge L$ is unsatifiable.
	\item For each $\hat{L}\subsetneq L$ the formula $\hat{\varphi}\wedge\hat{L}$ is satisfiable.
	\item Let $\ell\in L$ then $\ell$ does not occur in $\hat{\varphi}$.
\end{itemize}

As resolution is refutationally complete we can conclude with the above properties that there is a refutation $\pi=\pi_1,\ldots,\pi_m$ for $\hat{\varphi}\wedge L$. Next we define the relation $\rightarrow$.
For this purpose let $a\in\pi$ and $b\in\pi$ then we define $a\rightarrow b$ if
there is a clause $c\in\pi$ and a variable $x$ such that $b=\mathcal{R}_x(a,c)$. We denote the transitive closure of $\rightarrow$ by $\rightarrow^\ast$. 
Now we can see that for each $\ell\in L$ we have $\ell\rightarrow^\ast\pi_m$---otherwise we would get a contradiction to the minimality of $L$.

Subsequently, we construct a derivation $\pi'$ from $\pi$. For $1\leq i\leq m$ we
define $\pi_i'$ as:
\[\pi'_i\defeq\pi_i\cup\{\neg \ell\mid \ell\in L,\ell\rightarrow^\ast\pi_i\}\]
As the only possibility for a resolution with a pivot in $L$ are resolutions with unit clauses from $L$ we can see that each clause in $\pi'$ is either a clause in $\hat{\varphi}\wedge L$, 
the result of resolving on previous elements of $\pi'$ or the copy of a previous element of $\pi'$---%
note to simplify the work with indices we use copies of clauses.
Finally, we can see that $\pi'_m=\neg \ell_1\vee\ldots\vee\neg \ell_m$. 
\end{proof}

\begin{lemma}\label{FindDerivation}
Let $\tau$ be an arbiter assignment and $\rho$ a universal assignment. 
If $\varphi^\rho\wedge\tau^\rho$ is unsatisfiable then a subclause of $\neg\tau^\rho$ is derivable
from $\Phi$ in the $\forall$Exp+Res calculus.
\end{lemma}
\begin{proof}
As $\varphi^\rho\wedge\tau^\rho$ is unsatisfiable there is a term $\hat{\tau}\subseteq\tau$ such that
$\varphi^\rho\wedge\hat{\tau}^\rho$ is unsatisfiable and such that for each term $t\subset\hat{\tau}$ the formula $\varphi^\rho\wedge\ t^\rho$ is satisfiable. 
This means we can apply Lemma~\ref{Aux:PropositionalResult}, which shows that $\neg\hat{\tau}$ is derivable by resolution from $\varphi^\rho$. As $\varphi^\rho$ is obtained by instantiating the matrix $\matrixf$ by the universal assignment $\rho$, we can derive $\neg\hat{\tau}$ in $\forall$Exp+Res from $\Phi$.
\end{proof}

Before we now use the above results to show that for each clause that is added to the arbiter solver we can derive a subclause, we will show that this clause is well-defined,

\begin{lemma}
If $\neg\matrixf\wedge\psi_{\mathit{Def}}\wedge\tau$ is satisfied by an assignment $\sigma$ then 
the formula $\matrixf\wedge\matrixf_A\wedge\tau\wedge\restr{\sigma}{U}$ is unsatisfiable.
\end{lemma}
\begin{proof}
If $\neg\matrixf\wedge\psi_{\mathit{Def}}\wedge\tau$ is satisfied by an assignment $\sigma$, then
$\matrixf\wedge\psi_{\mathit{Def}}\wedge\tau\wedge\sigma$ is unsatisfiable. As $\tau$ uniquely determines the assignment for the arbiter variables and $\tau$, $\restr{\sigma}{U}$ and $\psi_{\mathit{Def}}$ unquietly determine the assignment for the existential variables also the formula $\matrixf\wedge\psi_{\mathit{Def}}\wedge\tau\wedge\restr{\sigma}{U}$ is unsatisfiable. 
This implies that the formula $\matrixf\wedge\matrixf_A\wedge\psi_{\mathit{Def}}\wedge\tau\wedge\restr{\sigma}{U}$ is unsatisfiable.
We can now apply Lemma~\ref{DefinitionRemoval} and conclude that 
$\matrixf\wedge\varphi_A\wedge\tau\wedge\restr{\sigma}{U}$ is unsatisfiable.
\end{proof}

The above lemma implies that the core extraction (line~\ref{alg1:getCore}) is well-defined.

\begin{restatable}{proposition}{propBasicAlgDerivableClauses}
\label{Prop:BasicDerivableClauses}
For each clause $C$ added to the arbiter solver by Algorithm~\ref{alg:BasicAlgorithm} (line~\ref{alg1:addclause}), a clause $C' \subseteq C$ can be derived from $\Phi$ in $\forall$Exp+Res.
\end{restatable}

\begin{proof}
By applying the Lemmata~\ref{CorrespondanceArbiterAssignmentAndAnnotatedLiterals} and~\ref{FindDerivation} we can conclude that
a subclause of $\neg\tau^{\restr{\sigma}{U}}$ is derivable from $\Phi$.
As $\neg\tau^{\restr{\sigma}{U}}$ is a subclause of $\neg\tau$ we have proven the lemma.
\end{proof}

\begin{theorem}\label{basicalg:com}
If Algorithm~\ref{alg:BasicAlgorithm} returns {\sc false} for the DQBF $\Phi$ then $\Phi$ is false.
\end{theorem}
\begin{proof}
If the algorithm returns {\sc false} then the set $\mathcal{C}$ of clauses in the arbiter solver is unsatisfiable.
By Proposition~\ref{Prop:BasicDerivableClauses} for each $C \in \mathcal{C}$ we can derive a clause $C' \subseteq C$ subsuming $C$ in $\forall$Exp+Res, so there is a $\forall$Exp+Res refutation of~$\FFF$.
As $\forall$Exp+Res is sound~\cite{BeyersdorffBCSS19}, this shows that $\FFF$ is false.
\end{proof}
Finally, Algorithm~\ref{alg:BasicAlgorithm} terminates since at most one arbiter variable is introduced for each existential variable and assignment of its dependency set in the first phase, and there is a limited number of clauses on arbiter variables that can be introduced in the second phase.
In combination with Theorem~\ref{basicalg:cor} and Theorem~\ref{basicalg:com}, we obtain the following result.
\begin{corollary}
Algorithm~\ref{alg:BasicAlgorithm} is a decision procedure for DQBF.
\end{corollary}
%

\subsection{Combining Definition Extraction with CEGIS}\label{Section:AdvancedAlgorithm}
Discounting SAT calls, the running time of Algorithm~\ref{alg:BasicAlgorithm} is essentially determined by the number of assignments of a dependency set for which the corresponding existential variable is not defined: it introduces an arbiter variable for each such assignment in the first phase, and the number of iterations in the second phase is bounded by the number of arbiter assignments.
As a result, even a single existential variable that is unconstrained and has a large dependency set causes the algorithm to get stuck enumerating universal assignments.

A key insight underlying the success of counter-example guided solvers
for QBF~\cite{JanotaKMC16,Janota18,Tentrup19} is that it is typically overkill to perform complete expansion of universal variables.
Instead, they incrementally refine Skolem functions by taking into account universal assignments that pose a problem for the current solution candidate.\footnote{In these QBF solvers, Skolem functions are typically only indirectly represented by trees of formulas (\emph{abstractions}) that encode viable assignments.}

Following this idea, we now present an improved algorithm (Algorithm~\ref{alg:AdvancedAlgorithm}) in the style of Counter-Example Guided Inductive Synthesis (CEGIS)~\cite{JhaS17}.
It integrates the two phases of Algorithm~\ref{alg:BasicAlgorithm} into a single loop.
In each iteration, it first tries to find definitions for existential variables in terms of their dependency sets and the arbiter variables (\Call{FindDefinitions}{}).
The algorithm then proceeds to a validity check of the definitions under the current arbiter assignment (\Call{CheckArbiterAssignment}{}).
A key difference to Algorithm~\ref{alg:BasicAlgorithm} is that we may not have a definition for each variable at this point.
In this case, we can simply leave the existential variable
unconstrained in the SAT call except for arbiter clauses $\varphi_A$
(and \emph{forcing clauses} $\varphi_F$, which we discuss later). In
the implementation, we limit the SAT solver's freedom to generate
counterexamples by substituting a default value or a heuristically
obtained ``guess'' for the Skolem function. Here, any function on variables
from the dependency set can be used without affecting correctness, one
only has to make sure that counterexamples are not repeated to guarantee termination.

If a counterexample~$\sigma$ is found, procedure \Call{CheckArbiterAssignment}{} returns it to the main loop.
Otherwise, (line~\ref{alg2:consistency}), we have to check whether the SAT call in line~\ref{alg2:validitycheck} returned UNSAT because a model has been found, or whether there is an inconsistency in the formula $\varphi_A \land \varphi_F$ comprised of arbiter and forcing clauses under the current arbiter assignment~$\tau$.
The procedure \Call{CheckConsistency}{} either finds that the model is consistent, in which case Algorithm~\ref{alg:AdvancedAlgorithm} returns {\sc true}, or else computes an assignment~$\sigma$ of the universal variables as a counterexample.
If \Call{CheckArbiterAssignment}{} returns {\sc false}, the main loop resumes in line~\ref{alg2:conflictanalysis} with a call to \Call{AnalyzeConflict}{}.

To see what this procedure does, let us first consider the simple case in which the counterexample~$\sigma$ only contains an assignment of universal variables that was returned by the consistency check.
Then, the existential assignment $\rho_\exists = \emptyset$ is empty, the for-loop is skipped and no new arbiter variables are introduced (line~\ref{alg2:newarbiters}), and the procedure only tries to further simplify the failed arbiter assignment $\rho_A$ in line~\ref{alg2:arbiterminimize}, before adding its negation to the arbiter solver.

Now assume $\rho_\exists$ is nonempty but the case distinction in the body of the for-loop between lines~\ref{alg2:forbegin} and \ref{alg2:forend} always leads to line~\ref{alg2:notforced}.
Then $\mathit{notforced} = \rho_\exists$ and the procedure \Call{NewArbiters}{} creates new arbiter variables~$A'$ and clauses~$\varphi_A'$ for each existential variable $e \in \funDom(\rho_\exists)$ and the universal counterexample $\sigma_\forall$ (restricted to the dependency set $D(e)$ in each case).
Since these arbiter variables determine the assignment of the existential variables in $\funDom(\rho_\exists)$ under~$\sigma_\forall$, we can replace $\rho_\exists$ with the assignment $\rho_A' \defeq \{ e^\xi \in A' \:|\: e \in \rho_\exists\} \cup \{ \neg e^\xi \in A' \:|\: \neg e \in \rho_\exists \}$
(line~\ref{alg2:arbiterassignmentextension})
and conclude that $\varphi \land \varphi_A \land \varphi_F$ is unsatisfiable under the assignment $\rho_A \cup \rho_A' \cup \sigma_\forall$, which only assigns arbiter variables and universal variables.
A clause forbidding the arbiter assignment~$\rho_A \cup \rho_A'$ can now be added as before.

Finally, let us turn to the general case, which includes entailment checks for each existential literal $\ell \in \rho_\exists$ in the minimized counterexample.
These checks are added to reduce the number of new arbiter variables created.
If the literal~$\ell$ is entailed by the assignment~$\sigma_\forall \land \tau$, we add further literals from~$\tau$ to the failed arbiter assignment~$\rho_A$ (if necessary) to ensure that~$\ell$ is entailed by $\sigma_\forall \land \rho_A$.
No arbiter variable has to be introduced for $\var(\ell)$ in this case. Otherwise, if $\neg \ell$ is entailed by $\sigma_\forall \land \tau$, then the counterexample is spurious since $e = \var(\ell)$ must be assigned the opposite way under $\sigma_\forall \land \tau$ by any Skolem function.
To enforce this in the next iteration, the algorithm adds a \emph{forcing clause}~$C$ encoding the implication $\sigma_\forall \land \tau \rightarrow \neg \ell$ (which can be further strengthened by restricting $\sigma_\forall$ to the dependency set of $e$) to $\varphi_F$.
It also sets a flag $\mathit{oppositeForced}$, which causes \Call{AnalyzeConflict}{} to exit with {\sc true} instead of adding new arbiter variables.

If \Call{AnalyzeConflict}{} returns {\sc true},  Algorithm~\ref{alg:AdvancedAlgorithm} proceeds to the next iteration of its main loop with the same arbiter assignment~$\tau$ but additional forcing clauses.
Otherwise, \Call{AnalyzeConflict}{} returns {\sc false} after adding a clause to the SAT solver $\mathit{arbiterSolver}$, and \Call{FindNewArbiterAssignment}{} is called to determine a new arbiter assignment~$\tau$ that satisfies all previously added clauses.
Algorithm~\ref{alg:AdvancedAlgorithm} terminates either when it
discovers a  model or when it cannot find a new arbiter assignment.
\begin{algorithm}[h!]
\caption{Solving DQBF by Definition Extraction (CEGIS Version)}\label{alg:AdvancedAlgorithm}
\begin{algorithmic}[1]
	\vspace{0.5em}

\Procedure{SolveByDefinitionExtractionCEGIS}{$\FFF$}
	\LineComment{$\FFF=\forall u_1\ldots\forall u_n\exists e_1(D_1)\ldots\exists e_m(D_m).\matrixf$}
	\LineComment{$A$: arbiter variables, $\psi_\mathit{Def}$: definitions, $\varphi_A$: arbiter clauses}
	\State$A\gets\emptyset,\,\psi_\mathit{Def}\gets\emptyset,\,\varphi_A\gets\emptyset$
        \State $\varphi_F \gets \emptyset$ \Comment{forcing clauses}
	\State $\tau\gets\emptyset$\Comment{arbiter assignment}
	\State $\mathit{arbiterSolver}\gets\textsc{SatSolver}(\emptyset)$
	\Loop
		\State$\psi_{\mathit{Def}}\gets\Call{FindDefinitions}{\{ e \in E \:|\: e \text{ undefined} \}, \matrixf\wedge\varphi_A\wedge\varphi_F}$
		\State$\mathit{modelValid},\sigma \gets\Call{CheckArbiterAssignment}{\tau}$
		\If{$\mathit{modelValid}$}
			\State\Return{\sc true}
                        \EndIf
                \LineComment{$\sigma$ is a counterexample} 
		\If{$\Call{AnalyzeConflict}{\sigma}$} \label{alg2:conflictanalysis}
                \LineComment{forcing clauses have been added to $\varphi_F$}
                \State \textbf{continue}
                \EndIf
		\If{not $\Call{FindNewArbiterAssignment}$}
			\State\Return{\sc false}
		\EndIf
	\EndLoop		
\EndProcedure

\Statex
    
\Procedure{CheckArbiterAssignment}{$\tau$}
	\State $\mathit{checker}\gets\textsc{SatSolver}(\neg\matrixf\wedge\psi_{\mathit{Def}}\wedge\varphi_F\wedge\varphi_A)$ \label{alg2:validitycheck}
	\If{$\mathit{checker}.\Call{solve}{\tau}$}
        \State$\sigma \gets\mathit{checker}.\Call{values}{E \cup U}$
       	\State\Return{\sc false}, $\sigma$
	\Else
        \State$\mathit{isConsistent},\sigma\gets\Call{CheckConsistency}{\varphi_A\wedge\varphi_F,\tau}$ \label{alg2:consistency}
		\If{$\mathit{isConsistent}$}
			\State\Return{\sc true}, $\emptyset$
                \Else
                        \LineComment{$\sigma \in [U]$ is such that $\varphi_A \land \varphi_F \land \tau \land \sigma$ is unsatisfiable}
                        \State\Return{\sc false}, $\sigma$
		\EndIf
	\EndIf
\EndProcedure
\Statex
\Procedure{FindNewArbiterAssignment}{\null}
	\If{$\mathit{arbiterSolver}.\Call{Solve}$}
		\State$\tau\gets\restr{\mathit{arbiterSolver}.\Call{getModel}}{A}$
		\State\Return{\sc true}
	\EndIf
	\State\Return{\sc false}
\EndProcedure
\algstore{advancedalgorithm}
\end{algorithmic}
\end{algorithm}

\begin{algorithm}[h!]
\begin{algorithmic}[1]
\algrestore{advancedalgorithm}
\Procedure{AnalyzeConflict}{$\sigma$}
\State $\sigma_{\forall} \gets \restr{\sigma}{U}$ \Comment{$\sigma_\forall$ assigns \emph{all} universal variables}
\State $\rho \gets\Call{getCore}{\matrixf\wedge\varphi_A\wedge\varphi_F, \sigma \wedge\tau}$\label{alg2:core}
\State $\rho_\exists\gets\restr{\rho}{E}$, $\rho_A\gets\restr{\rho}{A}$
\State $\mathit{notForced}\gets\emptyset$  \Comment{collect literals $\ell \in \rho_\exists$ that are not implied}
\State $\mathit{oppositeForced} \gets$ {\sc false}
\State $\psi \gets \matrixf\wedge\varphi_A\wedge\varphi_F$
	\For{$\ell\in\rho_\exists$} \label{alg2:forbegin}
		\If{$ \psi \wedge \sigma_\forall\wedge \tau \models \ell$}
			\State $\rho\gets\Call{getCore}{\psi,\sigma_\forall\wedge\tau\wedge \neg \ell}$
			\State $\rho_A\gets\rho_A\cup\restr{\rho}{A}$ \Comment{add reason for $\ell$ to failed arbiter assignment $\rho_A$}
        \ElsIf{$\psi\wedge\sigma_\forall\land \tau\models\neg \ell$}
			\State $\varphi_F\gets\varphi_F\wedge\Call{getForcingClause}{\psi,\sigma_\forall\wedge\tau,\neg\ell}$
			\State $\mathit{oppositeForced}\gets\textsc{true}$
		\Else
			\State $\mathit{notForced}\gets\mathit{notForced}\cup\{\ell\}$ \label{alg2:notforced}
		\EndIf
                \EndFor
                \If{$\mathit{oppositeForced}$}
                \State \Return {\sc true}
                \EndIf
                	\LineComment{no literal was forced to the opposite polarity}\label{alg2:forend}
                        \State$ \varphi_A',A'\gets\Call{newArbiters}{\mathit{notForced},\sigma_\forall}$\label{alg2:newarbiters}
                        \State $\varphi_A \gets \varphi_A \land \varphi_A'$
                        \State $\rho_A\gets\rho_A\wedge\Call{setAssignment}{A',\rho_\exists}$\label{alg2:arbiterassignmentextension}
		\State$\rho_A \gets \restr{\Call{getCore}{\psi,\rho_A \land \sigma_\forall}}{A}$ \label{alg2:arbiterminimize}
		\State$\mathit{arbiterSolver}.\Call{addClause}{\neg\rho_A}$ \label{alg2:addclause}	
                \State \Return {\sc false}
\EndProcedure

\end{algorithmic}
\end{algorithm}
We now prove that Algorithm~\ref{alg:AdvancedAlgorithm} is a decision procedure for DQBF.
As some of the required proofs are similar to related proofs in the previous section we will not give all the details.

As in Section~\ref{Section:BasicAlgorithm}, $A$ denotes a set of arbiter variables and $\varphi_A$ denotes the associated set of arbiter clauses.
Additionally, $\psi$ denotes some formula with variables in $U\cup E\cup A$.

Several proofs given in this section build on related proofs given in Section~\ref{Section:BasicAlgorithm}.
As several properties can be proven similarly as related properties in Section~\ref{Section:BasicAlgorithm} we will
sometimes give no proof, respectively only a sketch of a proof.
\begin{definition}[Forcing Clause]
Let $\ell$ be an existential literal, $\psi$ a formula with $\var(\psi)\subseteq U\cup E\cup A$
and let $\sigma$ be a (partial) assignment for $U\cup A$. 
We say that $\ell$ is \emph{forced} by $\sigma$ in $\psi$ if $\psi\wedge\sigma\wedge\neg \ell$ is unsatisfiable.
If $\ell$ is forced by $\sigma$ then $\neg\restr{\sigma}{D(\var(\ell))\cup A}\vee l$ is a \emph{forcing clause}.
\end{definition}
In particular, if a literal $\ell$ is forced by an assignment $\sigma$ in a formula $\varphi$ then $\varphi\wedge\sigma\vDash \ell$ holds. %
\begin{remark}
We can see that if a literal is forced then it keeps being forced after the addition of a clause.
\end{remark}
\begin{lemma}\label{Forcing_Model_Aux}
Let $\neg p\vee \ell$ be a forcing clause for $\ell$ in $\psi$. 
Then there is no model $F$ for $\qprefix\exists A(\emptyset).\psi$ such that:
\begin{itemize}
	\item For every arbiter literal $a\in p$ the model function $F_{\var(a)}$ satisfies $a$.
	\item For some assignment $\sigma\in[D(\var(\ell))]$ with $\sigma\vDash\restr{p}{U}$ the assignment $F_{\var(\ell)}(\sigma)$ satisfies $\neg \ell$.
\end{itemize}
\end{lemma}
\begin{proof}
As $\neg p\vee \ell$ is a forcing clause we know that there is a partial assignment $\rho$ for $U\cup A$ such that $p=\restr{\rho}{D(\var(\ell))\cup A}$ and $\psi\wedge\rho\wedge\neg \ell$ is unsatisfiable. 
Now let $F$ be a model with the above properties and $\sigma$ as in the second property. 
Furthermore, let $\sigma'\in[U]$ such that $\sigma\subseteq\sigma'$ and such that for each $x\in\funDom(\rho)\cap(U \setminus D(\var(\ell)))$ we have $\sigma'(x)=\rho(x)$. 
As $\sigma'\cup F(\sigma')$ satisfies $\rho\wedge\neg \ell$ we can conclude that $\sigma'\cup F(\sigma')$ falsifies
$\psi$. This means that $F$ is not a model.
\end{proof}
\begin{lemma}\label{Model_correspondance_force}
Let $\Psi$ be a DQBF with $\Psi\defeq\qprefix\exists A(\emptyset).\psi$ and let $\neg p\vee \ell$ be a forcing clause in $\psi$. Moreover, let $\Psi'\defeq\qprefix\exists A.\psi\wedge(\neg p\vee \ell)$. Then $F$ is a model for $\Psi$ if and only if $F$ is a model for $\Psi'$.
\end{lemma}
\begin{proof}
If $F$ is a model for $\Psi'$ then it is necessarily also a model for $\Psi$. Thus, it suffices to only consider the other direction of the proof.
Now assume that $F$ is a model for $\Psi$ but not for $\Psi'$ and show a contradiction.
As $F$ is not a model for $\Psi'$ there must be an assignment $\sigma$ for the universal variables such that
$\sigma\cup F(\sigma)$ falsifies $\neg p\vee \ell$. Now we can see that $F$ satisfies the properties given in Lemma~\ref{Forcing_Model_Aux}. This means that $F$ is not a model for $\Psi$. This proofs the lemma.
\end{proof}
An immediate consequence of the above lemma is the following corollary.
\begin{corollary}
Let $C=\neg p\vee \ell$ be a forcing clause in $\psi$. 
Then $\qprefix\exists A.\psi\wedge C$ is true if and only if $\qprefix\exists A.\psi$ is true.
\end{corollary}
The above results imply that forcing clauses can be added to the matrix of a DQBF without changing its models. In particular, the resulting DQBF has the same truth value.
\begin{lemma}\label{CorrectAux}
Let $\FFF'$ be the DQBF $\FFF'\defeq\qprefix\exists A(\emptyset).\matrixf\wedge\varphi_A$. 
Then, $\FFF$ is true if and only if $\FFF'$ is true.
\end{lemma}
\begin{proof}
We assume that $\Phi$ is true. This means that there is a model function $f_e$ in $\Phi$ for each existential variable $e$. We now construct model functions $f'_e$ for $\Phi'$ by:
\[f'_x\defeq
\begin{cases}
	f_x &\text{for }x\in E\\
	f_e(\sigma) &\text{for } x=e^{\sigma}
\end{cases}
\]
Obviously the model given above satisfies $\matrixf$ for all assignments to the universal variables.
Now assume that there is a universal assignment $\rho$ such that a clause $\ell^{\sigma}\vee\neg\sigma\vee\neg \ell$ in $\varphi_A$ is falsified under the above model. Let $e=\var(\ell)$ then it suffices to consider the case $\restr{\rho}{D(e)}=\sigma$. 
Moreover, we can assume that $f_e$ satisfies $\ell$ under $\sigma$---otherwise the clause would be satisfied. 
The above assumptions imply that $f_{e^\sigma}\vDash\ell^\sigma$. But this yields a contradiction to the initial assumption.
As the above clauses were arbitrary we can conclude that $\matrixf\wedge\varphi_A$ is satisfied by the model under each universal assignment.

Now assume that $\FFF'$ is true. We can see that by restricting a model of $\FFF'$ to the variables in $E$ we get a model for $\FFF$. This proves the equivalence.
\end{proof}

By combining the above results we can derive the following corollary.
\begin{corollary}\label{Forcing_equiv_Aux}
Let $C_1,\ldots,C_k$ be clauses such that for each index $i$, the clause $C_i$ is a forcing clause in
$\matrixf\wedge\varphi_A\wedge\bigwedge_{1\leq j<i}C_j$. 
Then the DQBF $\qprefix\exists A(\emptyset).\matrixf\wedge\varphi_A\wedge\bigwedge_{1\leq i\leq k}C_i$ is true if and only if $\Phi$ is true.
\end{corollary}


Subsequently, we will give a mild generalization of Lemma~\ref{TrueDQBFAux}. 
As this Lemma can be proven similarly as Lemma~\ref{TrueDQBFAux} we do not give a proof.
\begin{lemma}\label{CEGIS_SAT_CONDITION}
Assume there is a linear ordering for $E$, denoted by $<_E$. 
If there is a set $E'\subseteq E$ such that
there is a formula $\psi_e$ with $\var(\psi_e)\subseteq D(e)\cup\{x\in E\mid D(x)\subseteq D(e),x<_E e\}$ such that 
$\neg\varphi\wedge\bigwedge_{e\in E}(e\leftrightarrow\psi_e)$ is unsatisfiable then the DQBF $\FFF$ is true.
\end{lemma}

\begin{theorem}\label{CAGIS:cor}
If Algorithm~\ref{alg:AdvancedAlgorithm} returns \textsc{true} for a DQBF $\Phi$ then $\Phi$ is true.
\end{theorem}
\begin{proof}
We assume that the algorithm returns \textsc{true} and show that the DQBF $\FFF$ is true.
We know that we have a set $A$ of arbiter variables, a set $\varphi_A$ of arbiter clauses, a set $\varphi_F$ of forcing clauses and a set $E'\subseteq E$ such that
each $e\in E'$ has a definition $\psi_e$ in $\matrixf\wedge\varphi_A\wedge\varphi_F$ by $D(e)\cup A$.
Let $\psi_{\mathit{Def}}\defeq\bigwedge_{e\in E'}(e\leftrightarrow\psi_e)$ and $\FFF'\defeq\qprefix\exists A(\emptyset).\matrixf\wedge\varphi_A\wedge\varphi_F$.
By Corollary~\ref{Forcing_equiv_Aux}, we know that $\FFF$ is true if and only if $\FFF'$ is true.
As the algorithm returns {\sc true}, we know that $\neg\matrixf\wedge\varphi_A\wedge\varphi_F\wedge\psi_{\mathit{Def}}$ is unsatisfiable.
By Lemma~\ref{CEGIS_SAT_CONDITION} and the above property we know that $\FFF'$ is true.
Thus, $\FFF$ is true as well.
\end{proof}

\begin{lemma}\label{Correspondance_Main_Lemma}
Let $\varphi_F$ be a set of forcing clauses, $\sigma$ be an assignment for universal variables and $\tau\in[A]$. If $\matrixf\wedge\varphi_A\wedge\varphi_F\wedge\tau\wedge\sigma$ is unsatisfiable then we can derive a subclause of $\neg\tau^\sigma$ from $\FFF$ in $\forall$Exp+Res.
\end{lemma}
\begin{proof}
If $\matrixf\wedge\varphi_A\wedge\tau\wedge\sigma$ is unsatisfiable then the result follows by Lemma~\ref{CorrespondanceArbiterAssignmentAndAnnotatedLiterals} and Lemma~\ref{FindDerivation}.
We can now assume that $\matrixf\wedge\varphi_A\wedge\tau\wedge\sigma$ is satisfied by an assignment $\rho$.
Now there has to be a forcing clause $(\neg p\vee\ell)$ that is falsified by $\rho$ such that the set of forcing clauses $\hat{\varphi}_F$ that were introduced before $(\neg p\vee\ell)$, is satisfied by $\sigma$---note $(\neg p\vee\ell)$ is a forcing clause for $\matrixf\wedge\varphi_A\wedge\hat{\varphi}_F$.
As $\rho$ falsifies $\neg p$ it satisfies $p$, as $\rho$ additionally satisfies $\tau$ we can conclude that the arbiter literals in $p$ are contained in $\tau$
Subsequently, $v_i,\ldots,v_k$ shall denote the universal literals in $p$ and $\ell_1^{\sigma_1},\ldots,\ell_l^{\sigma_l}$ the arbiter literals in $p$. Moreover, $\FFF'$ shall denote the DQBF $\FFF'\defeq\qprefix\exists A(\emptyset).\matrixf\wedge\varphi_A\wedge\hat{\varphi}_F$. Because, of the above property $\FFF'$ cannot have a model $F$ with $F\vDash\bigwedge\ell_i^{\sigma_i}$ and $F(\restr{\rho}{U})\vDash\ell$. 
On the other hand because of Lemma~\ref{Forcing_Model_Aux} there can also be no model with $F\vDash\bigwedge\ell_i^{\sigma_i}$ and $F(\restr{\rho}{U})\vDash\neg\ell$. This means that if $\FFF'$ has a model than it has to falsify $\bigwedge\ell_i^{\sigma_i}$. We know that adding forcing clause to a matrix does not change models. Thus, if $\FFF''\defeq\qprefix\exists A(\emptyset).\matrixf\wedge\varphi_A$ has a model $F$ than it has
to falsify $\bigwedge\ell_i^{\sigma_i}$. This means that there is some minimal subset $S$ of $\{\ell_1^{\sigma_1},\ldots,\ell_l^{\sigma_l}\}$ such that $\qprefix\exists A(\emptyset).\matrixf\wedge\varphi_A\wedge\bigwedge S$ is false. Now let $\varphi_A^S$ be the subset of $\varphi_A$ with arbiters in $S$. We show that $\hat{\FFF}\defeq\qprefix\exists A(\emptyset).\matrixf\wedge\varphi_A^S\wedge\bigwedge S$ is false. We assume the opposite. In this case there is a model for the formula. But such a model can easily be extended to a model for $\FFF''$, which yields a contradiction. 

Next we argue that $\tilde{\FFF}\defeq\qprefix.\matrixf\wedge\bigwedge\{\neg\sigma\vee\ell\mid\ell^\sigma\in S\}$ is false.
Again, assume the opposite, i.e.\ there is a model $F$. We now show that that this model can be extended to a model $F'$ for $\hat{\FFF}$. Now let $\ell^\sigma\in S$, we set $F_{\var(\ell^\sigma)}=F_{\var(\ell)}(\sigma)$.
We can see that $F'$ is a model for $\hat{\FFF}$, but this is a contradiction. Moreover, by using similar arguments as above we can show that for any subset $S'$ of $S$ the formula $\qprefix.\matrixf\wedge\bigwedge\{\neg\sigma\vee\ell\mid\ell^\sigma\in S'\}$ is true.

Now let $C$ denote the set of clause that represent the full expansion of $\tilde{\FFF}$---in the sense of $\forall$Exp+Res.
As $\tilde{\FFF}$ is false, there is a $\forall$Exp+Res proof. This means there is a refutation by resolution for $C$. We can now on the one hand see that $S\subseteq C$ and on the other hand, because of the minimality of $S$ that removing any element from $S$ from $C$ would make $C$ satisfiable.
By Lemma~\ref{FindDerivation} we can now conclude that we can derive $\overline{S}\defeq\{\neg x\mid x\in S\}$  by resolution from $C$. From this we can conclude that $\overline{S}$ is derivable from $\FFF$ in $\forall$Exp+Res. As $\overline{S}$ is a subset of $\neg\tau$ this proves the lemma.
\end{proof}

As in the case of Algorithm~\ref{alg:BasicAlgorithm}, the correctness of {\sc false} answers for Algorithm~\ref{alg:AdvancedAlgorithm} follows from a correspondence with $\forall$Exp+Res derivations.
Before we prove that each clause that is added to the arbiter solver correponds to a clause that can be derived by $\forall$Exp+Res we show that the clauses that are added to the arbiter solver are well-defined.

\begin{lemma}
The core extraction in line~\ref{alg2:core} is well defined.
\end{lemma}
\begin{proof}
We have to show that $\psi\wedge\rho_A\wedge\sigma_\forall$ is indeed unsatisfiable.
There are two cases to consider. First assume that $\varphi_A\wedge\varphi_F$ was inconsistent under $\tau$ and $\sigma_\forall$. The inconsistency means that $\varphi_A\wedge\varphi_F\wedge\tau\wedge\sigma_\forall$ is unsatisfiable. This means that the core extraction is well-defined. In the second case the formula 
$\neg\matrixf\wedge\psi_{\mathit{Def}}\wedge\varphi_F\wedge\varphi_A\wedge\tau$ was satisfied by $\sigma\cup\tau$.
This means that $\matrixf\wedge\psi_{\mathit{Def}}\wedge\varphi_F\wedge\varphi_A\wedge\varphi_A'\wedge\sigma\wedge\restr{\rho}{A}$ is unsatisfiable---where $\varphi_A'$ denotes the arbiter clauses that were introduces in this step. By Lemma~\ref{DefinitionRemoval} the formula $\matrixf\wedge\varphi_F\wedge\varphi_A\wedge\varphi_A'\wedge\sigma\wedge\restr{\rho}{A}$ is unsatisfiable. 
Moreover, we can argue that $\matrixf\wedge\psi_{\mathit{Def}}\wedge\varphi_F\wedge\varphi_A\wedge\varphi_A'\wedge\restr{\sigma}{U}$ is unsatisfiable---the existential assignments in $\sigma$ are either forced or ruled out by the new arbiter clauses. 
Finally, by similar arguments as in previous proofs, also removing $\psi_{\mathit{Def}}$ does preserve unsatisfiability.
This shows that also in this case the core extraction is well-defined.
\end{proof}

\begin{restatable}{proposition}{CEGISPropositionDerivableClauses}
\label{Prop:CEGISDerivableClauses}
For each clause $C$ added to the arbiter solver by Algorithm~\ref{alg:AdvancedAlgorithm}
(line~\ref{alg2:addclause}), a clause $C'\subseteq C$ can be derived from $\Phi$ in $\forall$Exp+Res.
\end{restatable}
\begin{proof}
Let $\neg\rho_A$ be a clause that is added to the arbiter solver. Then we know that $\matrixf\wedge\matrixf_A\wedge\matrixf_F\wedge\rho_A\wedge\sigma$ for some universal assignment $\sigma$ is unsatisfiable. By applying Lemma~\ref{Correspondance_Main_Lemma} we can conclude that we can derive a subset of $\neg\rho_A$.
\end{proof}
%
\begin{restatable}{theorem}{CEGISCompletenessTheorem}
\label{CAGIS:com}
If Algorithm~\ref{alg:AdvancedAlgorithm} returns \textsc{false} for a DQBF $\Phi$ then $\Phi$ is false.
\end{restatable}
\begin{proof}
If the algorithm returns \textsc{false} then the clauses in the $\mathit{arbiterSolver}$ are unsatisfiable.
Subsequently, we denote those clauses by $\mathcal{C}$.
By Proposition~\ref{Prop:CEGISDerivableClauses}, we know that we can derive for each clause $C\in\mathcal{C}$ a subclause. This means that those derivable clauses are unsatisfiable. As propositional resolution is refutationally complete we can derive the empty clause from those derivable clauses. 
This means that there is a $\forall$Exp+Res proof for $\Phi$. As $\forall$Exp+Res is sound, this means that $\Phi$ is false.
\end{proof}
%
Each iteration of Algorithm~\ref{alg:AdvancedAlgorithm} introduces new forcing clauses or forbids another arbiter assignment. Because there is a bound on the number of arbiter variables that can be introduced, the number of such clauses can be bounded as well, and the algorithm eventually terminates.
Together with Theorem~\ref{CAGIS:cor} and Theorem~\ref{CAGIS:com}, this gives rise to the following corollary.
\begin{corollary}
Algorithm~\ref{alg:AdvancedAlgorithm} is a decision procedure for DQBF.
\end{corollary}

\section{Experiments}
\label{sec:experiments}
We implemented Algorithm~\ref{alg:AdvancedAlgorithm} as described in the previous section in a prototype named \tool.\footnote{Available at \url{https://github.com/perebor/pedant-solver}.}
For definition extraction, it uses a subroutine from {\sc Unique}~\cite{Slivovsky20} that in turn relies on an interpolating version of {\sc MiniSat} \cite{EenS03} bundled with the {\sc ExtAvy} model checker~\cite{GurfinkelV14,VizelGM15}.
Further, {\sc CaDiCaL} is used as a SAT solver~\cite{BiereFazekasFleuryHeisinger-SAT-Competition-2020-solvers} (we also tested with {\sc CryptoMiniSAT}~\cite{SoosNC09} and {\sc Glucose}~\cite{AudemardS09} but saw no significant differences in overall performance).
\tool\ can read DQBF in the standard DQDIMACS format and output models in the DIMACS format.

The implementation incorporates a few techniques not explicitly mentioned in the
above pseudocode.
We identify \emph{unate} existential literals (a generalization of pure literals)~\cite{AkshayCGKS18}, which can be used in any model of a DQBF.
Moreover, we set a (configurable) default value for existential variables that applies when there is no forcing clause propagating a different value.
This is to limit the freedom of the SAT solver used in the validity check in coming up with counterexamples.
Moreover, when checking for definability of an existential variable, we use \emph{extended dependencies} that include existential variables with dependency sets that are contained in the dependencies of the variable that is checked.

For all experiments described below we use a cluster with Intel Xeon E5649 processors at 2.53~GHz running 64-bit Linux.

\subsection{Performance on Standard Benchmark Sets}
We compare \tool\ with other DQBF solvers on standard benchmark sets
in terms of instances solved within the timeout and their PAR2 score.\footnote{The Penalized Average Runtime (PAR) is the average runtime, with the time for each unsolved instance calculated as a constant multiple of the timeout.}
Specifically, we choose the solvers {\sc dCAQE}~\cite{TentrupR19}, {\sc iDQ}~\cite{FrohlichKBV14}, {\sc HQS}~\cite{GitinaWRSSB15}, and the recently introduced {\sc DQBDD}~\cite{Sic20}.
Both {\sc HQS} and {\sc DQBDD} internally use {\sc HQSPre}~\cite{WimmerSB19} as a preprocessor.
For {\sc dCAQE} and {\sc iDQ}, we call {\sc HQSPre} with a time limit of $300$ seconds (the time for preprocessing is included in the total running time).
By default, \tool\ is run without preprocessing.

The results are based on a single run with a time and memory limit of $1800$ seconds and $8$~GB, respectively, which are enforced using {\sc RunSolver}~\cite{Roussel11}.\footnote{Due to the heavy-tailed runtime distribution of DQBF solvers, run-to-run variance rarely affects the number of solved instances. However, PAR2 scores should be taken with a grain of salt and only used to compare orders of magnitude.}
We report results for two benchmark sets. The first---which we refer to as the ``Compound'' set---has been used in recent papers on {\sc HQS}~\cite{Ge-ErnstSW19}. 
It is comprised of instances encoding partial equivalence checking (PEC)~\cite{SchollB01,FrohlichKBV14,GitinaRSWSB13,FinkbeinerT14} and controller synthesis~\cite{BloemKS14}, as well as succinct DQBF representations of propositional satisfiability~\cite{BalabanovCJ14}.
Results are summarized in Table~\ref{tab:compound}.
\tool\ solved the most instances overall and for 4 out of 5 families (the ``Balabanov'' family being the exception), with {\sc DQBDD} coming in a close second.
The performance of \tool\ on the PEC instances in the ``Finkbeiner'' family is particularly encouraging.
\newcommand{\cdott}{\!\cdot\!} 
\begin{table}
\caption{\label{tab:compound} Results for the ``Compound'' benchmark set.}
\noindent\resizebox{\linewidth}{!}{%
\begin{tabular}[t]{@{}l@{}r@{/}lr@{/}lr@{/}lr@{/}lr@{/}l@{}}
\toprule
 & \multicolumn{2}{c}{\sc dCAQE} & \multicolumn{2}{c}{DQBDD} &
                                                           \multicolumn{2}{c}{HQS}
  & \multicolumn{2}{c}{\sc iDQ} & \multicolumn{2}{c}{\sc Pedant}  \\
\cmidrule(l{3pt}r{3pt}){2-3} \cmidrule(l{3pt}r{3pt}){4-5} \cmidrule(l{3pt}r{3pt}){6-7} \cmidrule(l{3pt}r{3pt}){8-9} \cmidrule(l{3pt}r{3pt}){10-11} 
Family(Total) & Sol & PAR2 & Sol & PAR2 & Sol & PAR2 & Sol & PAR2 & Sol & PAR2 \\
\midrule
Balabanov(34) & \textbf{21} & $1.5\cdott 10^{3}$ & 13 & $2.3\cdott 10^{3}$ & 19 & $1.8\cdott 10^{3}$ & \textbf{21} & $1.5\cdott 10^{3}$ & 13 & $2.3\cdott 10^{3}$\\
Biere(1200) & \textbf{1200} & $1.6\cdott 10^{-1}$ & 1197 & $9.0\cdott 10^{0}$ & \textbf{1200} & $6.4\cdott 10^{-2}$ & 1184 & $6.6\cdott 10^{1}$ & \textbf{1200} & $1.0\cdott 10^{-1}$\\
Bloem(461) & 85 & $2.9\cdott 10^{3}$ & 82 & $3.0\cdott 10^{3}$ & 82 & $3.0\cdott 10^{3}$ & 50 & $3.2\cdott 10^{3}$ & \textbf{98} & $2.9\cdott 10^{3}$\\
Finkbeiner(2000) & 32 & $3.5\cdott 10^{3}$ & 1999 & $1.1\cdott 10^{1}$ & 1799 & $3.9\cdott 10^{2}$ & 6 & $3.6\cdott 10^{3}$ & \textbf{2000} & $1.7\cdott 10^{0}$\\
Scholl(1116) & 568 & $1.8\cdott 10^{3}$ & 793 & $1.1\cdott 10^{3}$ & 676 & $1.4\cdott 10^{3}$ & 345 & $2.5\cdott 10^{3}$ & \textbf{854} & $8.7\cdott 10^{2}$\\
\midrule
All(4811) & 1906 & $2.2\cdott 10^{3}$ & 4084 & $5.5\cdott 10^{2}$ & 3776 & $7.9\cdott 10^{2}$ & 1606 & $2.4\cdott 10^{3}$ & \textbf{4165} & $4.9\cdott 10^{2}$\\
\bottomrule
\end{tabular}}
\end{table}

Next, we consider the instances from the DQBF track of QBFEVAL'20~\cite{PulinaS19}. 
Results are shown in Table~\ref{tab:eval20}. Here, \tool\ falls behind the other solvers, with the exception of {\sc iDQ}.
In particular, significantly fewer instances from the ``Kullmann'' and ``Tentrup'' families are solved.

For the autarky finding benchmarks in the ``Kullmann'' family~\cite{KullmannS19}, we noticed that most dependencies can be removed by preprocessing with the reflexive resolution-path dependency scheme~\cite{SlivovskyS16,WimmerSWB16}.
The resulting instances are much easier to solve for \tool, and models can still be validated against the original DQBFs.
In general, we found that preprocessing with {\sc HQSPre} can have both positive and negative effects on \tool.
The rightmost columns of Table~\ref{tab:eval20} show results when preprocessing is enabled.\footnote{With options \texttt{--resolution 1 --univ\_exp 0 --substitute 0}.}
Overall, performance is clearly improved, but fewer instances from the ``Bloem'' and ``Scholl'' families are solved.
In prior work, it was observed that preprocessing can destroy definitions~\cite{Slivovsky20}, and this appears to be the case here as well.

For the instances from the ``Tentrup'' family, we discovered that the performance of \tool\ is sensitive to which counterexamples are generated by {\sc CaDiCaL}.
With the right sequence of counterexamples, false instances can be refuted quickly, while otherwise the solver is busy introducing arbiter variables for minor variations of previously encountered cases.
Curiously, this also appears to be the case for true instances.
We believe that the algorithm can be made more robust against such
``adversarial'' sequences of counterexamples by achieving better generalization (see Section~\ref{sec:conclusion}).

\begin{table}
\caption{\label{tab:eval20} Results for the QBFEVAL'20 DQBF benchmark set.}
\noindent\resizebox{\linewidth}{!}{%
\begin{tabular}[t]{@{}l@{}r@{/}lr@{/}lr@{/}lr@{/}lr@{/}lr@{/}l@{}}
\toprule
 & \multicolumn{2}{c}{\sc dCAQE} & \multicolumn{2}{c}{DQBDD} & \multicolumn{2}{c}{HQS}
  & \multicolumn{2}{c}{\sc iDQ} & \multicolumn{2}{c}{\sc Pedant} & \multicolumn{2}{c}{\sc PedantHQ} \\
\cmidrule(l{3pt}r{3pt}){2-3} \cmidrule(l{3pt}r{3pt}){4-5} \cmidrule(l{3pt}r{3pt}){6-7} \cmidrule(l{3pt}r{3pt}){8-9} \cmidrule(l{3pt}r{3pt}){10-11} \cmidrule(l{3pt}r{3pt}){12-13}
Family(Total) & Sol & PAR2 & Sol & PAR2 & Sol & PAR2 & Sol & PAR2 & Sol & PAR2 & Sol & PAR2\\
\midrule
Balabanov(34) & \textbf{21} & $1.5\cdott 10^{3}$ & 13 & $2.3\cdott 10^{3}$ & 19 & $1.8\cdott 10^{3}$ & \textbf{21} & $1.5\cdott 10^{3}$ & 14 & $2.3\cdott 10^{3}$ & 13 & $2.4\cdott 10^{3}$\\
Bloem(90) & 31 & $2.4\cdott 10^{3}$ & 32 & $2.3\cdott 10^{3}$ & 33 & $2.3\cdott 10^{3}$ & 14 & $3.1\cdott 10^{3}$ & \textbf{37} & $2.2\cdott 10^{3}$ & 25 & $2.7\cdott 10^{3}$\\
Kullmann(50) & 35 & $1.1\cdott 10^{3}$ & \textbf{50} & $1.5\cdott 10^{1}$ & 41 & $6.9\cdott 10^{2}$ & \textbf{50} & $3.4\cdott 10^{0}$ & 34 & $1.3\cdott 10^{3}$ & 40 & $7.3\cdott 10^{2}$\\
Scholl(90) & 52 & $1.5\cdott 10^{3}$ & 78 & $4.9\cdott 10^{2}$ & 77 & $5.3\cdott 10^{2}$ & 15 & $3.0\cdott 10^{3}$ & \textbf{82} & $3.3\cdott 10^{2}$ & 65 & $1.2\cdott 10^{3}$\\
Tentrup(90) & 77 & $5.5\cdott 10^{2}$ & \textbf{84} & $2.8\cdott 10^{2}$ & 78 & $5.1\cdott 10^{2}$ & 17 & $2.9\cdott 10^{3}$ & 15 & $3.0\cdott 10^{3}$ & \textbf{84} & $2.9\cdott 10^{2}$\\
\midrule
All(354) & 216 & $1.4\cdott 10^{3}$ & \textbf{257} & $1.0\cdott 10^{3}$ & 248 & $1.1\cdott 10^{3}$ & 117 & $2.4\cdott 10^{3}$ & 182 & $1.8\cdott 10^{3}$ & 227 & $1.4\cdott 10^{3}$\\
\bottomrule
\end{tabular}}
\end{table}

\subsection{Distribution of Defined Existential Variables}
The main design goal for \tool\ was to create a solver that benefits from unique Skolem functions given by propositional definitions. We thus expect \tool\ to do well on instances where a large proportion of existential variables is defined.
Figure~\ref{fig:unique} shows the distribution of defined existential variables (i.e., unique Skolem functions) as computed by {\sc Unique}~\cite{Slivovsky20}.
\begin{figure}
  \centering
  \includegraphics[width=\linewidth]{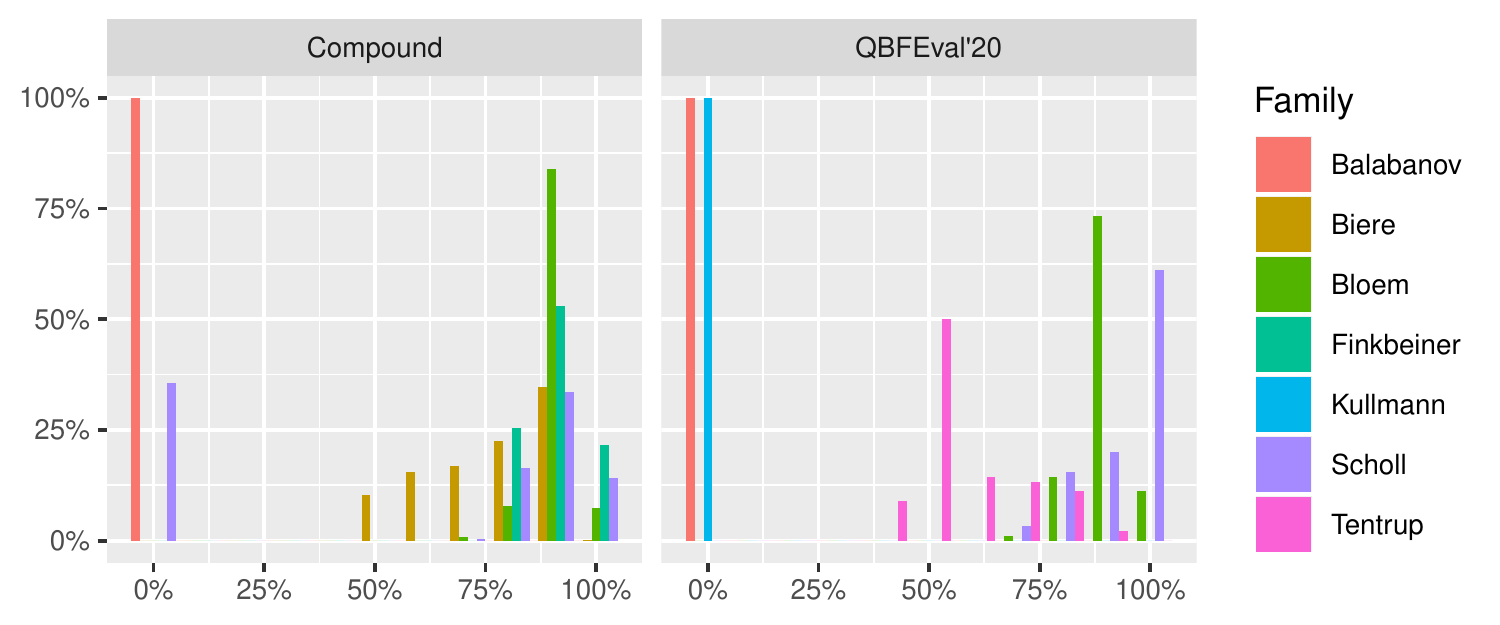}
  \caption{Distribution of defined variables by benchmark set and family. For a given percentage $x_0$ on the x-axis, the y-axis shows the fraction of instances from each benchmark family for which $x_0$ percent of existential variables are defined. For example, the instances in the ``Balabanov'' family have no defined variables, while the fraction of defined variables for instances in the ``Finkbeiner'' family ranges from $75$\% to $100$\%.}\label{fig:unique}
\end{figure}
These definitions are also found by \tool\ without the introduction of arbiter variables.
Comparing Table~\ref{tab:compound} and Table~\ref{tab:eval20} with Figure~\ref{fig:unique}, we see that \tool\ performed better for instance families with a larger fraction of defined variables.
This makes sense: the fewer variables are undefined, the fewer arbiter variables need to be introduced.

\subsection{Solution Validation}
When running \tool\ without preprocessing (the default), we had it trace and output models in DIMACS format.
We implemented a simple workflow for validating these models in Python~3 using the {\sc PySAT} library~\cite{imms-sat18}.
First, a simple syntactic check is performed to make sure the encoding of each Skolem function only mentions variables in the dependency set of the corresponding variable.
Then, a SAT solver is used to verify that substituting the model~$\psi$ for existential variables in the matrix~$\varphi$ of the input DQBF is valid, by testing for each clause~$C \in \varphi$ whether $\psi \land \neg C$ is unsatisfiable~(cf. Lemma~\ref{TrueDQBFAux}).
In this manner, we are able to validate models for all $648$ true DQBFs in the two benchmark sets that were solved by \tool\ without preprocessing. The maximum validation time was $237$ seconds, with a mean of $4.3$~seconds and a median of $0.5$ seconds.

The current validation process is intended as a proof of concept. Since models constructed by \tool\ are circuits, we plan to support the AIGER format~\cite{Biere-FMV-TR-11-2} in the near future, and provide a workflow along the lines of {\sc QBFCert}~\cite{NiemetzPLSB12}.


\section{Related Work}
\label{sec:related}
%
The DQDPLL algorithm lifts the CDCL algorithm to
DQBF~\cite{FrohlichKB12}. While CDCL solvers are free to assign
variables in any order, in DQBF a variable may be assigned only after
the variables in its dependency set have been assigned. 
Moreover, its assignment must not differ between branches in the search tree that agree on the assignment of the dependency set.
In DQDPLL, this is enforced by temporary \emph{Skolem clauses} that fix the truth value of a variable for a given assignment of its dependencies.
The solver {\sc dCAQE} lifts clausal abstraction from QBF to DQBF~\cite{TentrupR19}.
QBF solvers based on abstraction maintain a propositional formula for each quantifier level that characterizes eligible moves in the evaluation game.
These \emph{abstractions} are refined by forbidding moves that are known to result in a loss.
Abstractions are linked to each other through auxiliary variables that indicate which clauses are satisfied at different levels.
{\sc dCAQE} organizes variables in a dependency lattice that determines the order in which their abstractions may be solved.
This can lead to variables being assigned after variables that do not appear in their dependency sets, and additional consistency checks have to be applied to ensure that Skolem functions do not exploit such spurious dependencies. {\sc dCAQE} uses \emph{fork resolution} as its underlying proof system~\cite{Rabe17}.

Expansion of universal variables can be successively applied to transform a DQBF into a propositional formula that can be passed to a SAT solver~\cite{BubeckB06}.
In practice, the space requirements of fully expanding a DQBF are prohibitive.
This can be addressed by only expanding some universal variables, as well as considering only a subset of the clauses generated by expansion. Even though such approaches degenerate into full expansion in the worst case, they can be quite effective.
The solver {\sc iDQ}~\cite{FrohlichKBV14} successively expands a DQBF in a counterexample-guided abstraction refinement (CEGAR) loop.
Initially, universal variables in each clause are expanded separately.
Satisfiability of the resulting propositional formula is checked by a SAT solver.
If it is unsatisfiable, so is the original DQBF.
Otherwise, {\sc iDQ} checks whether any pair of literals with consistent annotations are assigned different truth values in the satisfying assignment. If there are no such literals, a model of the DQBF has been found.
Otherwise, clauses containing the corresponding clashing literals are further expanded. The system is inspired by the \emph{Inst-Gen} calculus, the proof system underpinning the First-Order solver {\sc iProver}~\cite{Korovin08}.
Originally designed for the effectively propositional fragment of first-order logic (EPR), {\sc iProver} also accepts DQBF as input.

The solver {\sc HQS} seeks to keep the memory requirements of expansion in check by operating on And-Inverter Graph (AIG) representations of input formulas~\cite{GitinaWRSSB15}.
It uses expansion alongside several other techniques to transform a DQBF into an equivalent QBF and leverage advances in QBF solving~\cite{WimmerKBS17,Ge-ErnstSW19}.
{\sc HQS} is paired with a powerful preprocessor named {\sc HQSPre} that provides an arsenal of additional simplification techniques~\cite{WimmerSB19}, including an incomplete but efficient method for refuting DQBF by reduction to a QBF encoding~\cite{FinkbeinerT14}.
{\sc HQSpre} is also used in the recently developed solver {\sc DQBDD}~\cite{Sic20}, which is similar to {\sc HQS} but relies on Binary Decision Diagrams (BDDs) instead of AIGs to represent formulas and perform quantifier elimination.

Evaluating DQBF is NEXPTIME complete~\cite{AzharPR01} in general, but some tractable subclasses have been identified in recent work~\cite{SchollJWG19,GanianPSS20}.


\section{Conclusion}
\label{sec:conclusion}
We presented a decision algorithm for DQBF that relies on definition
extraction to compute Skolem functions inside a CEGIS loop, and
evaluated it in terms of the prototype implementation \tool.
While the initial results are very promising, we see significant room for improvement and various directions to pursue in future research.
Generally, the approach works well when Skolem functions can be computed by definition extraction for a large fraction of existential variables without introducing too many arbiter variables.
During testing, we encountered multiple instances for which conflict analysis was occupied dealing with minor variations of a small number of counterexamples.
We believe that this is partly due to arbiter variables being introduced for \emph{complete} assignments of dependency sets.
Even if the assignment of some universal variables in the dependency set is irrelevant for a given counterexample, the newly introduced arbiter variables only deal with the counterexample as represented by the complete assignment, and each counterexample obtained by varying the assignment of irrelevant universal variables requires a new set of arbiter variables.
To avoid this, we plan to experiment with a variant of the algorithm that introduces arbiter variables for \emph{partial} assignments~\cite{Korovin08,FrohlichKBV14}.

A different approach to generalizing from counterexamples---one that does not require changes in the underlying proof system---is the use of machine learning.
By predicting the pattern common to a sequence of counterexamples, it is possible to deal with it wholesale and avoid an exhaustive enumeration~\cite{Janota18}.
Moreover, recent work on Boolean Synthesis demonstrates the viability of learning Skolem functions by sampling satisfying assignments~\cite{GoliaRM20}.

Finally, we plan to explore further applications of interpolation-based definition extraction within our algorithm.
Currently, its use is limited to existential variables that are defined by their dependency sets in the input DQBF, or are undefined only in a small number of cases.
In addition to that, one could search for ``partial'' definitions
under assignments of the dependency set characterized by formulas, or
introduce definitions that are valid under assumptions~\cite{RabeS16}.

%
\subsubsection*{Acknowledgements}
Supported by the Vienna Science and Technology Fund (WWTF) under the grants ICT19-060 and ICT19-065, and the Austrian Science Fund (FWF) under grant W1255.

\bibliographystyle{splncs04}
\bibliography{SAT2021}


\end{document}